\documentclass{article}
\usepackage{ijcai19}

\usepackage{times}
\usepackage{soul}
\usepackage{url}
\usepackage[hidelinks]{hyperref}
\usepackage[utf8]{inputenc}
\usepackage[small]{caption}
\usepackage{graphicx}
\usepackage{amsmath}
\usepackage{booktabs}
\usepackage{amssymb}
\usepackage{amsthm}
\usepackage{paralist}
\usepackage[noend,linesnumbered]{algorithm2e}

\newtheorem{theorem}{Theorem}

\newcommand\numberthis{\addtocounter{equation}{1}\tag{\theequation}}
\newcommand{\E}{\mathbb{E}}
\newcommand{\detO}{\mathrm{det}}
\newcommand{\kk}{z}

\newcommand{\uv}{\tilde{V}_{\mathrm{LB}}}
\newcommand{\ov}{\tilde{V}_{\mathrm{UB}}}

\usepackage{tikz}
\usepackage{pgfplots}

\usepackage{todonotes}
\usepackage{xcolor}

\title{Compact Representation of Value Function \\ in Partially Observable Stochastic Games}

\author{
Karel Hor\'{a}k$^1$
\and
Branislav Bo\v{s}ansk\'{y}$^1$\and
Christopher Kiekintveld$^2$\And
Charles Kamhoua$^3$
\affiliations
$^1$Czech Technical University in Prague, FEE, Department of Computer Science\\
$^2$The University of Texas at El Paso, Computer Science Department\\
$^3$Army Research Laboratory, Network Security Branch
\emails
\{horak, bosansky\}@agents.fel.cvut.cz,
cdkiekintveld@utep.edu,
charles.a.kamhoua.civ@mail.mil
}

\begin{document}

\maketitle

\begin{abstract}
  Value methods for solving stochastic games with partial observability model the uncertainty about states of the game as a probability distribution over possible states. The dimension of this belief space is the number of states. For many practical problems, for example in security, there are exponentially many possible states which causes an insufficient scalability of algorithms for real-world problems. To this end, we propose an abstraction technique that addresses this issue of the curse of dimensionality by projecting high-dimensional beliefs to characteristic vectors of significantly lower dimension (e.g., marginal probabilities). Our two main contributions are (1) novel compact representation of the uncertainty in partially observable stochastic games and (2) novel algorithm based on this compact representation that is based on existing state-of-the-art algorithms for solving stochastic games with partial observability. Experimental evaluation confirms that the new algorithm over the compact representation dramatically increases the scalability compared to the state of the art.
\end{abstract}

\section{Introduction}

Partially Observable Stochastic Games (POSGs) are a very general model of dynamic multi-agent interactions under uncertainty.
This makes POSGs useful for modeling many types of problems where players have restricted information about the environment and they can dynamically react to other players based on limited observations of the actions of the opponents.  
Examples include patrolling games~\cite{Basilico2009,vorobeychik2014-icaps,basilico2016,kucera2018}, where a defender protects a set of targets against an attacker, and  pursuit-evasion~\cite{chung2011search} where a pursuer is trying to find and apprehend an evader.
Other types of security games~\cite{tambe2011,fang2015-ijcai,Fang16,nguyen2017-botnet} can be extended to more realistic dynamic settings if the defender can observe and react to some information about attackers during the game.  

The generality of POSGs comes at the cost of computational complexity. 
Computing optimal strategies in POSGs is highly intractable from the algorithmic perspective, even in the two-player zero-sum setting.
If we assume that both players have partial information about the environment, the players need to reason not only about their belief over possible states, but also about the belief the opponent has over the possible states, \emph{beliefs over beliefs}, and so on.
Restricting to subclasses of POSGs where this issue of \emph{nested beliefs} does not arise allows us to design and implement algorithms that are guaranteed to converge to optimal strategies~\cite{horak2017aaai,horak2019aaai}.
However, the scalability of the algorithms even for this case is limited.

One of the fundamental problems is the complexity of representing, updating, and reasoning about uncertainty over a potentially very large state space.  
In these POSGs (as in single-agent Partially Observable Markov Decision Processes (POMDPs)), beliefs are probability distributions over the possible states.
This is a~well-known disadvantage of these models, since memory and computation time grow rapidly due to the \emph{curse of dimensionality}. 
Taking a similar approach to previous work on POMDPs~(e.g., in \cite{roy2005finding,li2010improving,zhou2010solving}), we address this problem by introducing a compact representation of uncertainty in POSGs, and we develop a novel algorithm based on this representation that dramatically improves scalability.  

As a motivating domain, we consider a cybersecurity example where an attacker uses lateral movement actions to expand his control of a network without being detected.
The defender tries to observe the attacker and take actions to protect the network by reconfiguring honeypots. 
A perfect-information version of this problem was proposed in~\cite{kamdem2017markov}. 
However, a more realistic model is that the defender has a limited ability to observe the actions of the attacker using imperfect monitoring actions. 
This version of the lateral movement game with uncertainty can be modeled as a one-sided POSG~\cite{horak2017aaai}. 
To use previous algorithms to solve this POSG, the belief of the defender is defined over the possible subsets of resources that the attacker may currently control in the network.
This representation scales exponentially in the size of the network, so it is intractable for all but the smallest examples.
We return to our motivating domain in Section~4 where key steps of our novel algorithm are discussed specifically for this domain.

Our main technical contribution is replacing the representation of beliefs over the exponential number of possible states using a compact \emph{characteristic vector} that captures key information but reduces the dimensionality of the beliefs. 
For POMDPs, a similar compact representation has often been selected based on Principal Components Analysis~\cite{roy2005finding}.
Motivated by many uses of marginal probabilities in security games, we propose using the marginal probability of each resource being infected as a characteristic vector, instead of explicitly considering all possible subsets of infected resources. 
While this offers a path to scaling to much larger problems, it has consequences for the solution quality as well as the construction of the solution algorithm. Many components of state-of-the-art POSG solvers are based on manipulating the full belief distribution and we show how these can be redesigned to operate using the more compact summarized belief representations. 
We formally define the fixed-point equation, show that the value function in the compact representation is still a convex function, and that solving a compact representation of the game yields a lower bound to the value of the original game.
For our motivating domain, we experimentally demonstrate that our novel algorithm operating on the compact representation scales significantly better (in orders of magnitude) with a negligible loss in quality (less than $1\%$) compared to the state of the art algorithm. 


\section{One-sided POSGs}
A \emph{one-sided partially observable stochastic game}~\cite{horak2017aaai}, or one-sided POSG, is an imperfect-information two-player zero-sum infinite-horizon game with perfect recall represented by a tuple $(S,A_1,A_2,O,T,R)$.
The game is played for an infinite number of \emph{stages}.
At each stage, the game is in one of the states $s \in S$ and players choose their actions $a_1 \in A_1$ and $a_2 \in A_2$ simultaneously.
An initial state of the game is drawn from a probability distribution $b^0 \in \Delta(S)$ over states termed the \emph{initial belief}.
The one-sided nature of the game translates to the fact that while player~1 lacks detailed information about the course of the game, player~2 is able to observe the game perfectly (i.e., his only uncertainty is the action $a_1$ player~1 is about to take in the current stage).

The choice of actions determines the outcome of the current stage:
Player 1 gets an \emph{observation} $o \in O$ and the game transitions to a state $s' \in S$ with probability $T(o,s' \,|\, s,a_1,a_2)$, where $s$ is the current state.
Furthermore, player 1 gets a reward $R(s,a_1,a_2)$ for this transition, and player 2 receives $-R(s,a,a')$.
The rewards are discounted over time with discount factor $\gamma < 1$.

One-sided POSGs can be solved by approximating the optimal value function $V^*: \Delta(S) \rightarrow \mathbb{R}$ using a pair of value functions $\underline{V}$ (lower bound on $V^*$) and $\overline{V}$ (upper bound on $V^*$).
The algorithm proposed in~\cite{horak2017aaai} refines these bounds by solving a sequence of \emph{stage games}.
In each of these stage games, the algorithm focuses on deciding the optimal strategies of the players in this stage (i.e., $\pi_1 \in \Delta(A_1)$ for player~1 and $\pi_2: S \rightarrow \Delta(A_2)$ for player~2) while assuming that the play in the subsequent stages yields values represented by value functions $\underline{V}$ or $\overline{V}$, respectively.
Solving the stage games also defines the fixed point equation for $V^*$,
\begin{align*}
    & V^*(b) = HV^*(b) = \min_{\pi_2} \max_{\pi_1} \Big( \E_{b,\pi_1,\pi_2}[R] \ + \numberthis\label{eq:orig_fixpoint} \\
    & \qquad + \gamma \sum_{a_1, o} {\Pr}_{b,\pi_1,\pi_2}[a_1,o] \cdot V^*(\tau(b,a_1,\pi_2,o)) \Big) \ \text{,}
\end{align*}
where $\tau(b,a_1,\pi_2,o)$ denotes the Bayesian update of the belief of player~1 when he played $a_1$ and observed $o$.
Note that actions of player~2 who has complete information affect the observations player 1 receives and thus belief of player 1.

Each stage game is parameterized by the current belief $b \in \Delta(S)$ of player 1.
For piecewise-linear and convex $\underline{V}$ and $\overline{V}$, a stage game is solved using linear programming.
We show this linear program for $\underline{V}$ (represented as a point-wise maximum over a set $\Gamma$ of linear functions $\alpha_i: \Delta(S) \rightarrow \mathbb{R}$).

\begin{subequations}\label{lp:onesided}

\small
\begin{align}
    \min & \ V \\
    \text{s.t.} & \sum_{a_2} \pi_2(s \wedge a_2) = b(s) & \forall s \label{lp:onesided:pi-sum} \\
                & V \geq \sum_{s, a_2} \pi_2(s \wedge a_2) R(s,a_1,a_2) + \gamma \sum_{o} V_{a_1 o} \!\!\!\!\!\!\!\!\! & \forall a_1 \label{lp:onesided:br} \\
                & b^{a_1 o}(s') = \sum_{s,a_2} \pi_2(s \wedge a_2) T(o,s'|s,a_1,a_2) \!\!\!  & \forall a_1, o, s' \label{lp:onesided:b-prime} \\
                & V_{a_1 o} \geq \sum_{s'} \alpha_i(s') b^{a_1 o}(s') & \forall a_1, o, \alpha_i \label{lp:onesided:subgame-alpha} \\
                & \pi_2(s \wedge a_2) \geq 0 & \forall s, a_2 \label{lp:onesided:pi-strategy}
\end{align}
\normalfont
\end{subequations}

In this linear program, player~2 seeks a strategy $\pi_2$ for the current stage of the game to minimize the utility $V$ of player~1.
Here $\pi_2(s \wedge a_2)$ stands for the joint probability (ensured by constraints~\eqref{lp:onesided:pi-sum} and~\eqref{lp:onesided:pi-strategy}) that the current state of the game is $s$ and player~2 plays the action $a_2$.
Constraint~\eqref{lp:onesided:br} stands for player~1 choosing the best-responding action $a_1$.
Constraint~\eqref{lp:onesided:b-prime} expresses the belief in the subsequent stage of the game when $a_1$ was played by player~1 and observation $o$ was seen (multiplied by the probability of seeing that observation), and finally constraint~\eqref{lp:onesided:subgame-alpha} represents the value of $\underline{V}$ in the belief $b^{a_1 o}$.
Such a linear program can be then combined with point-based variants of the value-iteration algorithm, such as Heuristic Search Value Iteration (HSVI)~\cite{smith2004uai,horak2017aaai}.

A different approximation scheme is used for the upper bound $\overline{V}$ on $V^*$.
Instead of using the point-wise maximum over linear function, $\overline{V}$ is expressed using a set $\Upsilon = \lbrace (b^{(i)},y^{(i)}) \,|\, 1 \leq i \leq |\Upsilon| \rbrace$ of points.
The lower convex hull of this set of points is then formed to obtain the value of $\overline{V}(b)$,
\begin{equation}
    \overline{V}(b) = \!\! \min_{\lambda \in \mathbb{R}^{|\Upsilon|}_{\geq 0}} \left\lbrace \sum_{1 \leq i \leq |\Upsilon|} \!\!\!\! \lambda_i y^{(i)} \ |\ \mathbf{1}^T \lambda = 1, \!\! \sum_{1 \leq i \leq |\Upsilon|} \!\!\!\! \lambda_i b^{(i)} = b \right\rbrace \ \text{.} \label{eq:convex-hull}
\end{equation}
Constraint~\eqref{lp:onesided:subgame-alpha} can then be adapted to use this representation of $\overline{V}$ as detailed in~\cite{horak2016gamesec}.

\section{Compact Representation of $V^*$}
The dimension of the value function $V^*$ depends on the number of states, which can be very large in cases like our motivating domain where the number of states is exponential in the size of the network. 
We propose an abstraction scheme called \emph{summarized abstraction} to decrease the dimensionality of the problem by creating a simplified represenation of the beliefs over the state space.   

We associate each belief $b \in \Delta(S)$ in the game with a \emph{characteristic vector} $\chi^{(b)} = \mathbf{A} \cdot b$ (for some matrix $\mathbf{A} \in \mathbb{R}^{k \times |S|}$ where $k \ll |S|$) and we define an (approximate) value function $\tilde{V}^*: \mathbb{R}^k \rightarrow \mathbb{R}$ over the space of characteristic vectors.

The main goal is to adapt algorithms based on value iteration to operate over the more compact space $\mathbb{R}^{k}$ instead of the original belief space $\Delta(S)$. 
First, we adapt the fixed point equation~\eqref{eq:orig_fixpoint} for one-sided POSGs to work with compact $\tilde{V}^*$ (instead of original $V^*$) by allowing  player 2 to choose \emph{any} belief that is consistent with the given characteristic vector $\chi$.
\begin{align*}
    & \tilde{V}^*(\chi) = \tilde{H}\tilde{V}^*(\chi) = \min_{b|\mathbf{A}b=\chi} \min_{\pi_2} \max_{\pi_1} \Big( \E_{b,\pi_1,\pi_2}[R] \ + \numberthis\label{eq:abs_fixpoint} \\
    & \qquad + \gamma \sum_{a_1, o} {\Pr}_{b,\pi_1,\pi_2}[a_1,o] \cdot \tilde{V}^*(\chi(\tau(b,a_1,\pi_2,o))) \Big)
\end{align*}

Next, we show that value function computed using dynamic operator $\tilde{H}$ gives a valid lower bound estimate of the original game representation over the belief space.
\begin{theorem}
    $\tilde{V}^*(\chi^{(b)}) \leq V^*(b)$ for every $b \in \Delta(S)$.
\end{theorem}
\begin{proof}
    We prove the theorem by induction.
    Let $\tilde{V}_0$ be arbitrary and let $V_0$ of the original game be $V_0(b) = \tilde{V}_0(\chi^{(b)})$ (i.e., $\tilde{V}_0 \leq V_0$).
    
    Assume $\tilde{V}_t(\chi^{(b)}) \leq V_t(b)$ for every $b \in \Delta(S)$.
    Now $H(\tilde{V}_t \circ \chi) \leq H(V_t)$.
    The extra minimization over beliefs $b$ in $H\tilde{V}$ can only decrease the utility of the stage game and hence
    \begin{equation}
        \tilde{V}_{t+1}(\chi) = \tilde{H}\tilde{V}_t(\chi) \leq H(\tilde{V}_t \circ \chi) \leq HV_t(b) = V_{t+1}(b)
    \end{equation}
\end{proof}

Note that the equation~\eqref{eq:abs_fixpoint} can also be solved using linear programming.
Consider a lower bound $\uv$ on $\tilde{V}^*$ formed as a point-wise maximum over linear functions $\alpha_i(\chi) = (\mathbf{a}^{(i)})^T \chi + \kk^{(i)}$.
We modify the linear program~\eqref{lp:onesided} by considering that the belief $b$ is a variable (constrained by $\chi$),
\begin{subequations}\label{lp:abstracted}
\begin{align}
    & \text{constraints \eqref{lp:onesided:pi-sum}, \eqref{lp:onesided:br}, \eqref{lp:onesided:b-prime}, \eqref{lp:onesided:pi-strategy}} \\
    & \sum_{s} b(s) = 1 \label{lp:abstracted:b-sum}\\
    & \mathbf{A} \cdot b = \chi \label{lp:abstracted:chi}\\
    & b(s) \geq 0 & \forall s \ \text{,}
\end{align}
and we replace the constraint~\eqref{lp:onesided:subgame-alpha} to account for different representation of $\uv$ by
\begin{align}
    & \chi^{a_1 o} = \mathbf{A} \cdot b^{a_1 o} & \forall a_1, o \ \ \ \\
    & q^{a_1 o} = \mathbf{1}^T \cdot b^{a_1 o} & \forall a_1, o \ \ \ \\
    & V_{a_1 o} \geq (\mathbf{a}^{(i)})^T \cdot \chi^{a_1 o} + \kk^{(i)} \cdot q^{a_1 o} & \forall a_1, o, \alpha_i \ \text{,} \label{lp:abstracted:subgame-alpha}
\end{align}
\end{subequations}
where $q^{a_1 o}=\mathbf{1}^T \cdot b^{a_1 o}$ is the probability that $o$ is generated when player~1 uses action $a_1$.

Observe that the only constraints with non-zero right-hand side in this new linear program are constraints~\eqref{lp:abstracted:b-sum} and~\eqref{lp:abstracted:chi}.
This is critical in the following proof that $\tilde{V}^*$ is convex.

\begin{theorem}\label{thm:convex}
    Value function $\tilde{V}^*$ is convex.
\end{theorem}
\begin{proof}
    Consider a dual formulation of the linear program~\eqref{lp:onesided} updated according to equations~\eqref{lp:abstracted}.
    Since the only non-zero right-hand side terms in the primal are 1 and $\chi$, the objective of the dual formulation is $o(\chi)=\chi^T \cdot \mathbf{a} + \kk$.
    Moreover, this is the only place where the characteristic vector $\chi$ occurs.
    Hence the polytope of the feasible solutions of the dual problem is the same for every characteristic vector $\chi \in \mathbb{R}^k$ and $o(\chi)$ (after fixing variables $\mathbf{a}$ and $\kk$) forms a bound on the objective value of the solution for \emph{arbitrary} $\chi$.
    Since we maximize over all possible $o(\chi)$ in the dual, the objective value of the linear program (and also $H\uv$) is convex in the parameter $\chi$.
    
    Now, starting with an arbitrary convex (e.g., linear) $\uv^0$, the sequence of functions $\lbrace \uv^t \rbrace_{t=0}^\infty$, where $\uv^{t+1}=\tilde{H}\uv^t$, is formed only by convex functions.
    Therefore, the fixed point  $\tilde{V}^*$ is also a convex function.
\end{proof}

\subsection{HSVI Algorithm for Compact POSGs}\label{sec:algorithm}
The Algorithm~\ref{alg:hsvi} we propose for solving abstracted games is a modified version of the original heuristic search value iteration algorithm (HSVI) for solving unabstracted one-sided POSGs~\cite{horak2017aaai}.
The key difference here is that we use the value functions $\uv$ and $\ov$ (instead of $\underline{V}$ and $\overline{V}$) and we must have modified all parts of the algorithm to use the abstracted representation of the beliefs.

\begin{algorithm}[t]
\caption{HSVI algorithm for one-sided POSGs when summarized abstraction is used.}
\label{alg:hsvi}
\SetKwFunction{Explore}{Explore}
\SetKwProg{myproc}{procedure}{}{}
\DontPrintSemicolon
Initialize $\uv$ and $\ov$ to lower and upper bound on $\tilde{V}^*$ \label{alg:hsvi:init}\;
\While{$\ov(\chi^0) - \uv(\chi^0) > \epsilon$}{\label{alg:hsvi:termination}
  \Explore{$\chi^0,\epsilon,0$} \label{alg:hsvi:exploration}\;
}
\vspace{0.3em}
\myproc{\Explore{$\chi,\epsilon,t$}}{\label{alg:hsvi:explore}
  $(b,\pi_2) \gets $ optimal belief and strategy of player 2 in $\tilde{H}\uv(\chi)$ \label{alg:hsvi:follower-choice}\;
  $(a_1,o) \gets $ select according to heuristic \label{alg:hsvi:heuristic}\;
  $\chi' \gets \tau(\chi,a_1,\pi_2,o)$ \label{alg:hsvi:chi-prime}\;
  Update $\Gamma$ and $\Upsilon$ based on the solutions of $\tilde{H}\uv(\chi)$ and $\tilde{H}\ov(\chi)$ \label{alg:hsvi:update:pre} \;
  \If{$\ov(\chi') - \uv(\chi') > \epsilon\gamma^{-t}$}{
      \Explore{$\chi',\epsilon,R,t+1$} \label{alg:hsvi:recursion}\;
      Update $\Gamma$ and $\Upsilon$ based on the solutions of $\tilde{H}\uv(\chi)$ and $\tilde{H}\ov(\chi)$ \label{alg:hsvi:update:post} \;
  }
}
\end{algorithm}

First, we initialize bounds $\uv$ and $\ov$ (line~\ref{alg:hsvi:init}) to valid piecewise linear and convex lower and upper bounds on $\tilde{V}*$ (using sets $\Gamma$ of linear functions $\alpha_i=(\mathbf{a}^{(i)})^T\chi + \kk^{(i)}$ and $\Upsilon$ of points $(\chi^{(i)},y^{(i)})$).
Then, we perform a sequence of trials (lines~\ref{alg:hsvi:termination}--\ref{alg:hsvi:exploration}) from the initial characteristic vector $\chi^0=\mathbf{A} b^0$ until the desired precision $\epsilon > 0$ is reached.

In each of the trials, we first compute the optimal \emph{optimistic} strategy of player~2, which in this case is the selection of the belief $b$ and strategy $\pi_2$ (line~\ref{alg:hsvi:follower-choice}).
Next, we choose the action $a_1$ of player~1 and the observation $o$ (lines~\ref{alg:hsvi:heuristic}--\ref{alg:hsvi:chi-prime}) so that the excess approximation error $\ov(\chi')-\uv(\chi')-\epsilon\gamma^{-t}$ in the subsequent stage (where the belief is described by a characteristic vector $\chi'=\tau(\chi,a_1,\pi_2,o)$) is maximized.
If this excess approximation error is positive, we recurse to the characteristic vector $\chi'$ (line~\ref{alg:hsvi:recursion}).

Before and after the recursion the bounds $\uv(\chi)$ and $\ov(\chi)$ are improved using the solution of $\tilde{H}\uv(\chi)$ and $\tilde{H}\ov(\chi)$ (lines~\ref{alg:hsvi:update:pre} and~\ref{alg:hsvi:update:post}).
The update of $\ov$ is straightforward and a new point $(\chi,\tilde{H}\ov(\chi))$ is added to $\Upsilon$.
To obtain a new linear function to add to the set $\Gamma$, we use the objective function $o(\chi)=\chi^T \cdot \mathbf{a} + \kk$ (after fixing variables $\mathbf{a}$ and $\kk$) of the dual linear program to~\eqref{lp:abstracted} that forms a lower bound on $\tilde{H}\uv$ and $\tilde{V}^*$ (see the proof of Theorem~\ref{thm:convex} for more discussion).

\section{The Lateral Movement POSG}\label{sec:model}
We illustrate our approach using a cybersecurity scenario based on the critical problem of detecting and mitigating lateral movement by attackers in networks. The game is played on a directed acyclic graph $G=(V,E)$, where the vertices $V=\lbrace v_1, \ldots, v_{|V|} \rbrace$ are sorted in topological order.
The goal of the attacker is to reach vertex $v_{|V|}$ from the initial source of the infection $v_1$ by traversing the edges of the graph, while minimizing the cost to do so.

Initially, the attacker controls only the vertex $v_1$, i.e., the initial \emph{infection} is $I_0=\lbrace v_1 \rbrace$.
Then, in every stage of the game, the attacker chooses a directed path $P=\lbrace P(i) \rbrace_{i=1}^k$ (where $k$ is the length of the path) from any of the infected vertices to the target vertex $v_{|V|}$.
Unless the defender takes countermeasures, the attacker infects all the vertices on the path, including the target vertex $v_{|V|}$, and pays the cost of traversing each of the edges on the path,
\begin{equation}
    c_{-,P} = \sum_{P(i) \in P} C(P(i)) \ \text{,}
\end{equation}
where $C(P(i))$ is a cost associated with taking edge $P(i)$.

The defender tries to discover the attacker and increase his costs by deploying honeypots into the network.
The honeypot is deployed on an edge (denote this edge $h$) of the graph, and is able to detect if the attacker traverses that specific edge.
Furthermore, it also increases the cost for using the edge $h$ to $\overline{C}(h)$.
If the attacker observes that he has traversed a honeypot, he may decide to change his plan and therefore does not execute the rest of his originally intended path $P$.
The cost of playing a path $P$ against a honeypot placement $h$ is therefore
\begin{equation}
    c_{h,P} = \sum_{P(i) \in P_{\leq h}} C(P(i)) + \mathbf{1}_{h \in P} \cdot [\overline{C}(h) - C(h)]
\end{equation}
where $P_{\leq h}$ is the prefix of the path $P$ until the interaction with the honeypot edge $h$,
\begin{equation}
    P_{\leq h} =
    \begin{cases}
        P & h \not\in P \\
        \lbrace P(i) \rbrace_{i=1}^{\overline{i}} \text{ where } P(\overline{i})=h & \text{otherwise\ \ .}
    \end{cases}
\end{equation}
Since the attacker does not need to continue to execute his selected path $P$ (since the defender can reconfigure the position of the honeypot), the new infection $I_{h,P}$ becomes
\begin{equation}
    I_{h,P} = I \cup \left\lbrace v \,|\, (u,v) \in P_{\leq h} \right\rbrace \ \text{.}
    \label{eq:I-transition}
\end{equation}

The above problem can be formalized as a one-sided POSG where
\begin{compactitem}
    \item states are possible infections, i.e., $S=2^V$,
    \item actions of the defender are honeypot allocations, i.e., $A_1=E$,
    \item actions $A_2$ of the attacker are paths in $G$ reaching $v_{|V|}$,
    \item observations denote whether the defender detected the attacker (i.e., $h \in P$) or the attacker has reached the target $v_{|V|}$ while avoiding the detection, $O=\lbrace \detO, \neg\detO \rbrace$,
    \item transitions follow the equation~\eqref{eq:I-transition} and observation $\detO$ is generated iff the honeypot edge $h$ has been traversed,
    \item reward of the defender is the negative value of the cost of the attacker, i.e., $R(h,P)=-c_{h,P}$,
    \item discount factor of the game is $\gamma=1$, and
    \item initial belief $b^0$ of the game satisfies $b^0(I_0)=1$.
\end{compactitem}
Note that the original HSVI algorithm for one-sided POSGs has been defined and proved for discounted problems with $\gamma<1$.
In this particular case, however, we expect that the convergence properties translate even to the undiscounted case since the game is essentially finite (in a finite number of steps, \emph{all} vertices, including $v_{|V|}$, get infected and the game ends).

\subsection{Characteristic Vectors}
The number of states in the game is exponential in the number of vertices of the graph, $|S| = 2^{|V|-2} + 1$ (we consider that $v_1$ is always infected and we treat all states where $v_{|V|}$ is infected as a single terminal state of the game).
We propose to use marginal probabilities of a vertex being infected as characteristic vectors $\chi \in \mathbb{R}^{|V|}$, i.e.
\begin{equation}
    \chi^{(b)}_i = \sum_{I \in S \,|\, v_i \in I} b(I) \ \text{,}
\end{equation}
where $I$ corresponds to some subset of possibly infected vertices of the graph and $b(I)$ denotes the original belief over this subset of vertices being infected. 

\subsection{Value Function Representation}
The algorithm from Section~\ref{sec:algorithm} approximates $\tilde{V}^*$ using a pair of value functions, the lower bound $\uv$ and the upper bound $\ov$.
While we have discussed the representation of the lower bound (which is represented using a set $\Gamma$ of linear functions $\alpha_i(\chi)=(\mathbf{a}^{(i)})^T \chi + \kk^{(i)}$), the representation of the upper bound is more challenging.

In the original algorithm, the value function $\overline{V}$ is defined (see equation~\eqref{eq:convex-hull}) over the probability simplex $\Delta(S)$.
In that case, it suffices to consider $|S|$ points in $\Upsilon$ to define $\overline{V}$ for every belief.
In contrary, the space of characteristic vectors (i.e., marginal probabilities) is formed by a hypercube $[0,1]^{|V|}$ with $2^{|V|}$ vertices, which would make the straightforward point representation (using equation~\eqref{eq:convex-hull}) impractical.

\begin{figure}
    \centering
    \includegraphics{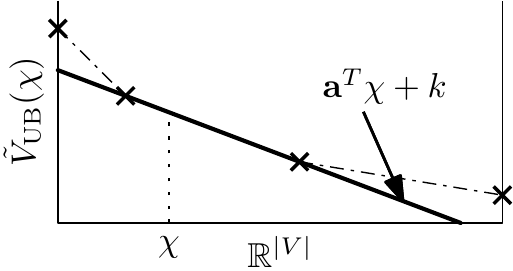}
    \caption{Dual interpretation of the projection on the convex hull.}
    \label{fig:projection}
\end{figure}
We can, however, leverage the fact that in this domain infecting an additional node can only decrease the cost to the target (and hence $\tilde{V}^*$ is decreasing).
Consider the dual formulation of the optimization problem~\eqref{eq:convex-hull}.
In this formulation, the projection of $\chi$ to the lower convex hull of a set of points is represented by the \emph{optimal} linear function $\mathbf{a}^T \chi + \kk$ defining a facet of the convex hull (see Figure~\ref{fig:projection}).
Since $\tilde{V}^*$ is decreasing in $\chi$, we can also enforce that $\mathbf{a}^T\chi + \kk$ is decreasing in $\chi$ (i.e., add the constraint $\mathbf{a} \leq 0$ to the dual formulation).
This additional constraint translates to a change of the equality $\sum_{1 \leq i \leq |\Upsilon|} \lambda_i \chi^{(i)} = \chi$ in the primal problem to an inequality.
\begin{equation}
    \ov(b) = \!\!\min_{\lambda \in \mathbb{R}^{|\Upsilon|}_{\geq 0}} \!\! \left\lbrace \sum_{1 \leq i \leq |\Upsilon|} \!\!\!\! \lambda_i y^{(i)} \ |\ \mathbf{1}^T \lambda = 1, \!\!\! \sum_{1 \leq i \leq |\Upsilon|} \!\!\!\! \lambda_i \chi^{(i)} \leq \chi \right\rbrace \label{eq:convex-hull-chi}
\end{equation}
Now it is sufficient that the set $\Upsilon$ contains just one point $(\chi^{(i)},y^{(i)})$ where $\chi^{(i)}=\mathbf{0}^{|V|}$ (instead of $2^{|V|}$ points) to make the constraint $\sum_{1 \leq i \leq |\Upsilon|} \lambda_i \chi^{(i)} \leq \chi$ satisfiable.

It is possible to adapt the constraint~\eqref{lp:abstracted:subgame-alpha} to use the representation from~\eqref{eq:convex-hull-chi}, similarly to the original (unabstracted) one-sided POSGs~\cite{horak2016gamesec}, and thus obtain linear program to perform the upper bound computation.

\subsection{Using Marginalized Strategies in Stage Games}
The linear program formed by modifications from equations~\eqref{lp:abstracted} still requires solving the stage game for the original, unabstracted problem.
In this section, we show that it is possible to avoid expressing the belief $b$ expicitly, and to compute the stage game directly using the characteristic vectors and marginalized strategies of the attacker.

First, we present the representation of the stage-game strategies of the attacker.
Instead of representing joint probabilities $\pi_2(I \wedge P)$ of choosing path $P$ in state $I$, we only model the probability $\tilde{\pi}_2(P)$ of choosing path $P$ aggregated over all states $S$.
Furthermore, we allow the attacker to choose the probability $\xi(P \wedge v_i)$ that vertex $v_i$ is infected while he opts to follow path $P$.
\begin{subequations}
\begin{align}
    & \sum_{P} \tilde{\pi}_2(P) = 1 \\
    & 0 \leq \xi(P \wedge v_i) \leq \tilde{\pi}_2(P) & \forall P, v_i \\
    & \tilde{\pi}_2(P) \geq 0 & \forall P
\end{align}
To ensure that the strategy represented by variables $\tilde{\pi}_2$ and $\xi$ is feasible it must be consistent with the characteristic vector $\chi$, where $\chi_i$ is the probability that the vertex $v_i$ is infected at the beginning of the stage.
\begin{align}
    & \sum_{P} \xi(P \wedge v_i) = \chi_i & \forall v_i
\end{align}
Furthermore, the path $P$ must start in an already infected vertex (denoted as $\mathrm{Pre}(P)$), i.e., the conditional probability $\Pr[\mathrm{Pre}(P) \in I \,|\, P]=1$ of $\mathrm{Pre}(P)$ being infected when path $P$ is chosen is 1.
Now, since $\xi(P \wedge v)$ is the joint probability, $\xi(P \wedge v)=\Pr[\mathrm{Pre}(P) \in I \,|\, P] \cdot \tilde{\pi}_2(P)$,
\begin{align}
    & \xi(P \wedge v) = \tilde{\pi}_2(P) & \forall P, v = \mathrm{Pre}(P) \ \text{.}
\end{align}

This representation of strategies of the attacker is sufficient to express the expected immediate reward of the strategy $\tilde{\pi}_2$, hence the constraint~\eqref{lp:onesided:br} can be changed to use the marginalized strategies,
\begin{align}
    & V \geq \sum_{P} \tilde{\pi}_2(P) c_{h,P} + \sum_{o} V_{h o} & \forall h \in E \ \text{.}
\end{align}

Importantly, we can also skip the computation of the belief $b^{h o}$ and compute the characteristic vector formed by the marginals $\chi^{h o}$ directly from the variables $\tilde{\pi}_2$ and $\xi$.
We now present the equation to compute the updated marginal $\chi^{h,\detO}$ given that the attacker has been detected while traversing the honeypot edge $h$.
\begin{equation}
    \chi^{h,\detO}_i = \!\!\!\!\!\!\!\!\!\!\!\!\!\!\!\sum_{P | h \in P \wedge P_{\leq (\cdot,v_i)} \subseteq P_{\leq h}}\!\!\!\!\!\!\!\!\!\!\!\!\!\!\! \tilde{\pi}_2(P) \qquad+ \!\!\!\!\!\!\!\!\!\sum_{P | h \in P \wedge P_{\leq (\cdot,v_i)} \not\subseteq P_{\leq h}}\!\!\!\!\!\!\!\!\!\!\!\!\!\!\! \xi(P \wedge v_i)
\end{equation}
The first sum stands for the probability that the attacker is detected while traversing edge $h$, but he infected $v_i$ beforehand.
The second sum represents the probability that the attacker used edge $h$ as well, but this time he has not infected $v_i$ using path $P$, however, the vertex $v_i$ has already been infected before starting to execute path $P$.

Analogously, we can obtain the probability $q^{h,\detO}$ that the attacker got detected while traversing edge $h$ as
\begin{equation}
    q^{h,\detO} = \sum_{P | h \in P} \tilde{\pi}_2(P) \ \text{.}
\end{equation}
\end{subequations}

We need not consider the subsequent stages where the attacker has not been detected (i.e., $\neg\detO$ observation has been generated) or the honeypot edge $h$ reaches the target vertex $v_{|V|}$.
In each of these cases, the target vertex has been reached and thus the value of the subsequent stage is zero.

\subsection{Initializing Bounds}
We now describe our approach to initialize bounds $\uv$ and $\ov$ (line~\ref{alg:hsvi:init} of Algorithm~\ref{alg:hsvi}).
Denote $C^*(v_i)$ and $\overline{C}^*(v_i)$ the shortest (i.e., cheapest) paths from $v_i$ to the target vertex $v_{|V|}$ when costs $C$ and $\overline{C}$, respectively, are used for all edges.

We initialize the lower bound $\uv$ using two linear functions $\alpha_1(\chi) = 0$ and $\alpha_2(\chi)=(\mathbf{a}^{(2)})^T \chi + \kk^{(i)}$ such that $\kk^{(i)}=C^*(v_1)$ and $\mathbf{a}^{(2)}_i=C^*(v_i)-C^*(v_1)$.

To initialize the upper bound $\ov$, we consider that exactly one vertex $v_i$ is infected and we consider the most expensive path $\overline{C}^*(v_i)$ from that vertex to the target $v_{|V|}$.
We use the set $\Upsilon=\lbrace (\chi^{(j)},y^{(j)}) \,|\, 1 \leq j \leq |V| \rbrace$ of points to initialize $\ov$ where
\begin{equation}
    \chi^{(j)}_i = \begin{cases}
        1 & i=j \\
        0 & i \neq j
    \end{cases}
    \qquad
    y^{(j)}=\overline{C}^*(v_j) \ \text{.}
\end{equation}

\section{Experimental Results}
In this section we experimentally evaluate the scalability and properties of our proposed abstraction technique based on the model from section~\ref{sec:model}.
Unless stated otherwise, we consider directed acyclic graphs generated from the Erd\H{o}s-R\'{e}nyi model with probability $p=0.5$ that each of the possible edges is included.
Furthermore, we make sure to include edges $(v_i,v_{i+1})$ for $1 \leq i \leq |V|-1$.
We set the costs $C(v_i,v_j)=j-i$ for every edge $(v_i,v_j)\in E$ when the honeypot is not deployed to $(v_i,v_j)$.
We do, however, apply significant penalty for traversing a honeypot edge where $\overline{C}(v_i,v_j)=j(j-i)$.
This setting is used to model layered networks commonly used in critical infrastructures~\cite{kuipers2006control}.
The attacker can either proceed to the subsequent layer (using edge $(v_i,v_{i+1})$), or use a shortcut (if available).
The costs $\overline{C}$ reflect the fact that the closer the attacker is to the target, the more secure the network is (i.e., he has to use a more expensive exploit to proceed).

All of the experiments have been evaluated on a machine with Intel i7-8700K and 32GB of RAM.
CPLEX 12.7.1 has been used to solve the linear programs used in the algorithms.

\subsection{Comparison with the Original Approach}
The main challenges of applying the original algorithm for solving one-sided partially observable stochastic games~\cite{horak2017aaai} are its memory requirements and the computational hardness of solving the stage games.
The algorithm explicitly reasons about each of the states of the game and the probability that the attacker chooses a given path in each of the states.
In this section, we demonstrate these challenges and illustrate the practical advantage of our proposed algorithm.

We used a set of randomly generated graphs (with varying number of vertices) and we attempted to solve these instances using both the original (unabstracted) approach and the algorithm we presented in this paper.
The parameters used for the original algorithm follow the parameters proposed in~\cite{horak2017aaai} while we modified the initialization of the bounds to make it valid for the undiscounted problem in question.
The target precision $\epsilon=0.1$ was used for both algorithms.

Figure~\ref{fig:unabstracted} shows the runtimes of the original algorithm (when applied to the unabstracted game) and our proposed approach.
We can observe that while for small instances, the original approach is competitive and can even outperform our proposed approach, for larger instances with large state spaces (the number of states in the game is up to $2^{|V|-2}+1$) the compact representation used in our approach turns out to be beneficial.
Moreover, the original approach has been unable to solve 50\% of the 11-vertex instances due to its memory requirements.
In contrast, our approach relying on the summarized abstraction based on marginal probabilities has not exceeded 1GB of memory consumption even when applied to the most challenging 17-vertex instances.
\newcommand\fmarksize{3pt}

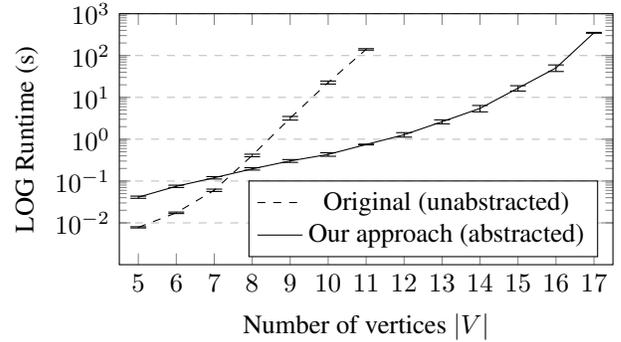
\begin{figure}
\centering
\begin{tikzpicture}
\begin{semilogyaxis}[ 
    width=0.95\linewidth,
    height=14em,
    xlabel={Number of vertices $|V|$},
    ylabel style={align=center},
    ylabel={LOG Runtime (s)},
    xmin=4.5, xmax=17.5,
    ymin=0.001, ymax=1000,
    xtick={5,6,7,8,9,10,11,12,13,14,15,16,17},
    ytick={0.01, 0.1, 1, 10, 100, 1000},
    legend pos=south east,
    legend style={font=\normalfont},
    ymajorgrids=true,
    grid style=dashed,
    scaled ticks=false, 
    mark size=\fmarksize,
    label style={font=\normalfont},
    tick label style={/pgf/number format/fixed, font=\normalfont} 
]

    \addplot[
    color=black,
    dashed,
    error bars/.cd,
        error bar style={color=black,solid},
    y dir=both, y explicit
    ]
    coordinates {
    (5.0, 0.00781188) +- (0.000253658,0.000253658)
    (6.0, 0.017405940594059) +- (0.0007289963261666,0.0007289963261666)
    (7.0,0.060207920792079) +- (0.0038861115998304,0.0038861115998304)
    (8.0,0.409940594059406) +- (0.0285329943110855,0.0285329943110855)
    (9.0,3.18382178217822) +- (0.318981536893921,0.318981536893921)
    (10.0,22.7026336633663) +- (1.87529621360055,1.87529621360055)
    (11.0,140.5449) +- (5.62986823831279,5.62986823831279)
    };

\addplot[
    color=black,
    error bars/.cd,
    y dir=both, y explicit
    ]
    coordinates {
    (5.0,0.041178217821782) +- (0.0024986554800374,0.0024986554800374)
    (6.0,0.074871287128713) +- (0.0047782143812587,0.0047782143812587)
    (7.0,0.119782178217822) +- (0.0079537739967942,0.0079537739967942)
    (8.0,0.19539603960396) +- (0.01262051567257,0.01262051567257)
    (9.0,0.302752475247525) +- (0.0235473667569883,0.0235473667569883)
    (10.0,0.432544554455445) +- (0.0414101642709914,0.0414101642709914)
    (11.0,0.746807692307692) +- (0.015969289781579,0.015969289781579)
    (12.0,1.27853465346535) +- (0.15819531697516,0.15819531697516)
    (13.0,2.61858415841584) +- (0.28638691948737,0.28638691948737)
    (14.0,5.45748514851485) +- (0.949392225017022,0.949392225017022)
    (15.0,16.608198019802) +- (2.42401332908958,2.42401332908958)
    (16.0,50.6150792079208) +- (9.07859864271665,9.07859864271665)
    (17.0,349.733732673267) +- (8.31833669324832,8.31833669324832)
    };

    \legend{Original (unabstracted), Our approach (abstracted)}
   
\end{semilogyaxis}
\end{tikzpicture}
\caption{Comparison of the runtime of the original (unabstracted) approach with the proposed one (using summarized abstraction). Confidence intervals mark standard error.}
\label{fig:unabstracted}
\end{figure}

\subsection{Abstraction Quality}
In this section, we focus on analyzing the abstraction quality.
In general, the summarized abstraction may lose important information needed to derive the optimal behavior.
We show, however, that with properly designed characteristic vectors the quality loss can be minimal.

The scalability issues of the original approach (see Figure~\ref{fig:unabstracted}) make it hard to determine the value $V^*(b^0)$ of the original, unabstracted game.
Therefore we present the quality analysis only for graphs with $5 \leq |V| \leq 10$ vertices where the original approach solved all instances considered.
In Table~\ref{tab:quality}, we present the empirical bound on the relative distance from the equilibrium of the unabstracted bound.
We compare the upper bound $\overline{V}(b^0)$ computed by the original approach for the unabstracted game with the lower bound $\uv(\chi^0)$ on the quality of the strategy when the abstraction is used.
We depict the maximum relative distance based on the 100 randomly generated instances for each size of a graph.
Observe that in all of the cases, the empirical upper bound on the relative distance (i.e., the worst-case possible quality loss due to abstraction) is below 1.0\%.
Note, that in all of the instances the quality of the strategy found by our novel algorithm is within the bounds of the original approach (i.e., the empirical bounds are likely to be overestimated).
\begin{table}
    \centering
    
    \bgroup
    \def\arraystretch{1.7}
    \setlength\tabcolsep{5pt}
    \begin{tabular}{|l||c|c|c|c|c|c|}
        \hline
        $|V|$ & 5 & 6 & 7 & 8 & 9 & 10 \\
        \hline
        $\frac{\overline{V}(b^0) - \uv(\chi^0)}{\uv(\chi^0)}$ (in \%) & 0.8 & 0.9 & 1.0 & 0.7 & 0.7 & 0.5 \\
        \hline
    \end{tabular}
    \egroup
    \caption{Empirical bound on the relative distance from the equilibrium of the unabstracted game based on 100 instances.}
    \label{tab:quality}
\end{table}

\section{Conclusions}
We focus on solving partially observable stochastic games (POSGs) and the representation of partial information in these games.
In the existing algorithms for solving subclasses of POSGs, the dimension of the belief space equals to the number of possible states.
This fact limits the scalability since both required memory and computation time grow rapidly.
We introduce a novel abstraction method and an algorithm for compact representation of the uncertainty in these POSGs. 
Our methodology is domain-independent and we demonstrate it on a motivating example from cybersecurity where the defender protects a computer network against an attacker who uses lateral movement in the network to reach the desired target. 
Experimental results show that our novel algorithm scales several orders of magnitude better compared to the existing state of the art with only negligible loss in quality (less than $1\%$). 
Our paper demonstrates practical aspects of algorithms for solving subclasses of POSGs and opens possibility to use similar compact representation for other domains.

\section*{Acknowledgment}
This research was supported by the Czech Science Foundation (grant no. 19-24384Y) and by the Army Research Laboratory and was accomplished under Cooperative Agreement Number W911NF-13-2-0045 (ARL Cyber Security CRA). The views and conclusions contained in this document are those of the authors and should not be interpreted as representing the official policies, either expressed or implied, of the Army Research Laboratory or the U.S. Government. The U.S. Government is authorized to reproduce and distribute reprints for Government purposes not with standing any copyright notation here on.

\bibliographystyle{named}
\bibliography{main}

\end{document}